\documentclass[12pt,notitlepage]{article}%
\usepackage{amsthm}
\usepackage{setspace}
\usepackage{eurosym}
\usepackage{subfigure}
\usepackage{xcolor}
\usepackage{amssymb}
\usepackage{graphicx}
\usepackage[colorlinks,linkcolor=links,citecolor=cites,urlcolor=MyDarkBlue]%
{hyperref}
\usepackage{amsfonts}
\usepackage{amsmath}
\usepackage{amstext}
\usepackage{appendix}
\usepackage{appendix}
\usepackage{mathpazo}
\usepackage{tikz}
\usepackage{verbatim}

\usetikzlibrary{trees}
\usetikzlibrary{matrix}
\usetikzlibrary{decorations.pathreplacing}
\usetikzlibrary{topaths,calc}
\usepackage{natbib}%
\providecommand{\U}[1]{\protect\rule{.1in}{.1in}}
\newtheorem{theorem}{Theorem}

\newtheorem{fact}{Fact}

\newtheorem{definition}{Definition}
\newtheorem{example}{Example}

\newtheorem{lemma}{Lemma}

\newtheorem{remark}{Remark}

\numberwithin{equation}{section}
\definecolor{MyDarkBlue}{rgb}{0,0.08,0.45}
\definecolor{cites}{HTML}{324b13}
\definecolor{links}{HTML}{1a663b}
\definecolor{MyLightMagenta}{cmyk}{0.1,0.8,0,0.1}
\hypersetup{
colorlinks,citecolor=blue,filecolor=black,linkcolor=blue,urlcolor=blue
}
\usepackage[top=1.25in,bottom=1.25in,left=1.25in,right=1.25in]{geometry}
\usepackage{setspace}
\usepackage{footmisc}
\usepackage{lipsum}
\begin{document}
\title{Two-sided matching with firms' complementary preferences}
\author{Chao Huang\thanks{Institute for Social and Economic Research, Nanjing Audit University. Email: huangchao916@163.com.}}
\date{}
\maketitle

\begin{abstract}
This paper studies two-sided many-to-one matching in which firms have complementary preferences. We show that stable matchings exist under a balancedness condition that rules out a specific type of odd-length cycles formed by firms' acceptable sets. We also provide a class of preference profiles that satisfy this condition. Our results indicate that stable matching is compatible with a wide range of firms' complementary preferences.
\end{abstract}



\section{Introduction}
The problem of two-sided matching considers how to match two kinds of agents, such as men and women, workers and firms, students and schools. A key solution concept is stable matching, which excludes incentives for agents to block the outcome. In different settings of two-sided matching, \cite{KC82} and \cite{R84} have shown that stable matchings always exist when there are no complementarities. This paper studies discrete many-to-one matching. Stable matchings exist in this setting when firms have the following \textbf{substitutable preferences} defined by \cite{RS90}.
\begin{equation}\label{sub}
\begin{aligned}
&\qquad\text{\emph{Any worker chosen by a firm from a set of available workers}}\\
&\text{\emph{would still be chosen when the available set \textbf{shrinks}.}}
\end{aligned}
\end{equation}
Each firm has a substitutable preference means that workers are not complements. For example, suppose a firm hires Bob from some available workers but would not hire Bob if the available set shrinks as Ana becomes unavailable. We would consider that Ana is a complement to Bob for this firm. Substitutable preferences are critical for matching practices. Market designers have provided numerous solutions based on the Deferred Acceptance algorithm (\citealp{GS62}) to real-life markets with substitutable preferences.\footnote{Substitutable preferences subsume unit demands and responsive preferences. There is vast, ongoing literature on real-life applications for two-sided matching; see \cite{H18} for a survey.}

On the other hand, complementarities are very common in many real-life markets. For example, firms commonly let workers cooperate in manufacturing and services, and thus a worker is usually a complement to some other worker. Complementarities in firms' preferences have been a critical issue for market design because a stable matching may not exist when there are complementarities. This paper studies many-to-one matching where firms have \textbf{complementary preferences} of the following form.\footnote{The term ``complementary preference'' has also been used in literature to mean preferences that violate (\ref{sub}). If we follow this convention, complementary preferences defined by (\ref{com}) may be called ``purely complementary preferences''.}

\begin{equation}\label{com}
\begin{aligned}
&\qquad\text{\emph{Any worker chosen by a firm from a set of available workers}}\\
&\text{\emph{would still be chosen when the available set \textbf{expands}.}}
\end{aligned}
\end{equation}
Each firm has a complementary preference means that workers are not substitutes. For example, suppose a firm hires Bob from some available workers but would not hire Bob if the available set expands as Ana becomes available. We would consider that Ana, in the latter case, has substituted Bob.

Substitutable and complementary preferences are proper assumptions for job markets of different features. For example, firms often have complementary preferences in a job market where scale returns are significant. Although there is a substitution effect between two workers, workers are not substitutes when the scale effect consistently exceeds the substitution effect. On the other hand, firms have substitutable preferences in a job market where scale returns are so insignificant that workers are not complements.

We provide a condition on firms' complementary preferences to guarantee the existence of a stable matching. We show that a stable matching always exists when firms have complementary preferences and firms' acceptable sets form a \textbf{balanced} matrix (\citealp{B70}). A 0-1 matrix is balanced if it has no square submatrix of odd order with exactly two 1s in each row and column.  Consider the following market with three firms $f_1,f_2,f_3$, and three workers $w_1,w_2,w_3$.

\begin{equation}\label{exam_in1}
\begin{aligned}
&f_1: \{w_1,w_2\}\succ\emptyset \qquad\qquad\qquad\qquad &w_1: &\quad f_1\succ f_3\\
&f_2: \{w_2,w_3\}\succ\emptyset \qquad\qquad\qquad\qquad &w_2: &\quad f_2\succ f_1\\
&f_3: \{w_1,w_3\}\succ\emptyset \qquad\qquad\qquad\qquad &w_3: &\quad f_3\succ f_2\\
\end{aligned}
\end{equation}
Each firm has a complementary preference, and no stable matching exists in this market. For example, if $f_1$ hires $w_1$ and $w_2$, leaving $w_3$ unemployed, then $w_2$ and $w_3$ would rather work for $f_2$. Similarly, neither $\{w_2,w_3\}$ matching to $f_2$ nor $\{w_1,w_3\}$ matching to $f_3$ is a stable outcome. The failure of the balancedness condition in (\ref{exam_in1}) is due to the following matrix on the left.

\medskip
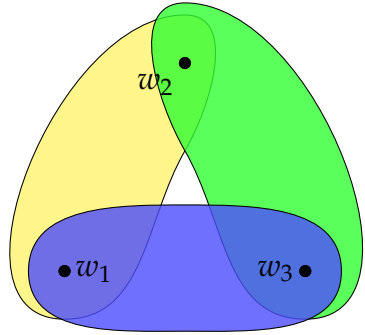
\begin{tabular}
[c]{c|cccc}
& $\{w_1,w_2\}$ & $\{w_2,w_3\}$ & $\{w_1,w_3\}$\\\hline
$w_1$ & $1$ & $0$ & $1$\\
$w_2$ & $1$ & $1$ & $0$\\
$w_3$ & $0$ & $1$ & $1$%
\end{tabular}
\begin{tikzpicture}[scale=0.8]
    \node (v1) at (0,0) {};
    \node (v2) at (2,3.46) {};
    \node (v3) at (4,0) {};
    \node (v4) at (2,2) {};
    \node (v5) at (2,-1) {};

    \begin{scope}[fill opacity=0.8]
    \filldraw[fill=yellow!70] ($(v2)+(0,0.8)$)
        to[out=0,in=60] ($(v4)$)
        to[out=240,in=0] ($(v1) + (0,-0.8)$)
        to[out=180,in=180] ($(v2) + (0,0.8)$);
    \filldraw[fill=green!80] ($(v2)+(0,1)$)
        to[out=180,in=120] ($(v4)$)
        to[out=300,in=180] ($(v3) + (0,-0.8)$)
        to[out=0,in=0] ($(v2) + (0,1)$);
    \filldraw[fill=blue!70] ($(v1)+(-0.6,0)$)
        to[out=90,in=180] ($(v4)+ (0,-0.9)$)
        to[out=0,in=90] ($(v3) + (0.6,0)$)
        to[out=270,in=0] ($(v5)$)
        to[out=180,in=270] ($(v1) + (-0.6,0)$);
    \end{scope}

    \foreach \v in {1,2,3} {
        \fill (v\v) circle (0.1);
    }

    \fill (v1) circle (0.1) node [right] {$w_1$};
    \fill (v2) circle (0.1) node [below left] {$w_2$};
    \fill (v3) circle (0.1) node [left] {$w_3$};

\end{tikzpicture}
\medskip

The balancedness condition can be described as a visualized constraint on firms' acceptable sets: The failure of balancedness in (\ref{exam_in1}) is attributed to the above odd-length cycle on the right. In particular, balancedness requires that every odd-length cycle formed by firms' acceptable sets has an acceptable set containing at least three workers of the cycle. An even-length cycle of any type does not cause the nonexistence of a stable matching. For example, readers can easily check that the market with four firms, four workers, and a similar structure as (\ref{exam_in1}) has a stable matching but does not admit a stable outcome when the numbers of firms and workers both increase to five.

Our main existence theorem allows two generalizations. First, by examining the structure of complementarities inside each firm, we can view each firm as a group of independent subsidiaries. We can then generalize the existence theorem by imposing balancedness only on firms' \textbf{primitive acceptable sets} that are relevant to the firm's independent subsidiaries; see Section \ref{Sec_gen}. The other generalization is due to the proof techniques for the existence theorem. We only need an additivity property of complementary preference to prove the existence theorem. Hence, stable matchings exist when balancedness is imposed on a more general class of firms' \textbf{additive preferences}; see Section \ref{Sec_add}.

Existing applications for matching theory mostly assume that firms (or schools, colleges, and hospitals) have substitutable preferences. Despite the great successes in market practice, substitutable preferences are also restrictive for ruling out all complements. In particular, whenever firm $f$ with a substitutable preference hires a worker $w$ from a set of available workers, $f$ should still hire $w$ if $f$'s available set shrinks; and thus, $f$ should hire $w$ when $w$ is the only available worker. This is in contrast to reality in which firms often hire at least a team of workers. In many real-life sectors, it is common for firms to hire several teams of workers, where members within a team are complements. Our existence theorem allows us to provide an application in which firms choose from teams of complementary workers. We follow \cite{H22} (henceforth, H22) to describe the structure of firms' acceptable sets as a ``technology tree''. Each technology of a technology tree demands a team of workers to implement. Firms choose to implement some technologies by recruiting the corresponding teams. We show that a technology tree of a specific structure induces balanced preference profiles. See Section \ref{Sec_app}.

\subsection{Related literature}
Complementarities have been regarded as the source of instabilities in matching theory. Positive results with complementarities of form (\ref{com}) are relatively rare. \cite{O08} showed that stable outcomes exist in a supply chain network where purchasing and selling contracts are complementary in a similar form as (\ref{com}). \cite{RY20} showed that stable matchings exist in a multilateral matching market with externalities where each agent has a complementary preference. Our market differs from this market because workers have unit demands in a two-sided matching market.

We study our problem using the method proposed in the author's recent work H22. H22 studied two-sided matching by investigating stable integral matchings of a fictitious market with ``divisible'' workers. This fictitious market is an instance of the continuum market of \cite{CKK19}, where a stable matching always exists according to their existence theorem. H22 investigated when a stable fractional matching can always be transformed into a stable integral one. We prove our results by applying H22's method to the market in which firms have complementary preferences.

Furthermore, our existence theorem generalizes H22's existence theorem when firms have complementary preferences. H22 showed that stable matchings exist if firms' demand type is totally unimodular.\footnote{See arXiv:2103.03418 for the latest version of H22. Earlier versions of H22 defined the firms' demand type differently and presented the condition as a unimodularity condition that is equivalent to the total unimodularity condition of the latest version; see footnote 21 of H22.} Total unimodularity of firms' demand type implies total unimodularity of firms' acceptable sets, which further implies balancedness of firms' acceptable sets. See Section \ref{Sec_uni}.

Our work also relates to studies on two-sided matching with complementarities, which include matching with couples (e.g., \citealp{KK05}, \citealp{KPR13}, \citealp{ABH14}, and \citealp{NV18}), matching with peer effects (e.g., \citealp{EY07} and \citealp{P12}), matching in large markets (e.g., \citealp{AH18}, \citealp{CKK19}, and \citealp{GK21}), and complementarities of other forms (e.g., \citealp{NV19}, and \citealp{H21a,H21b}).

The integer programming method of H22 and this paper also relates to the linear programming method for selecting stable matchings of specific properties (e.g., \citealp{V89}, \citealp{R92}, \citealp{RRV93}, \citealp{BB00}. \citealp{TS98} and \citealp{STQ06}), the integer programming methods for studying complexities of specific matching problems (\citealp{BMM14} and \citealp{ABM16}), and the fixed-point method that characterizes the existence of a stable matching (e.g., \citealp{A00}, \citealp{F03}, \citealp{EO04,EO06}, and \citealp{HM05}).

Complementarities in two-sided matching also relate to complementarities in markets with continuous monetary transfers and quasi-linear utilities (e.g., \citealp{DKM01}, \citealp{SY06}, \citealp{AWW13}, \citealp{HKNOW13}, \citealp{BK19}, and \citealp{TY19}). Firms and workers can bargain over continuous wages in the latter framework. H22 and this paper also allow generalizations in which firms and workers bargain over multidimensional contracts (\citealp{HM05}) that may specify wages, insurance, and retirement plans. \cite{E12} showed that firms and workers essentially bargain over one dimension when contracts are not complements and that \cite{HK10} exploit the framework with contracts in full generality. The latter paper's bilateral and unilateral substitutes conditions apply to cadet-branch matching (\citealp{S13} and \citealp{SS13}) and affirmative actions in school choice (\citealp{KS16}). H22 and this paper provide another opportunity to investigate the framework in which contracts may be complements.

The remainder of this paper is organized as follows. Section \ref{Sec_M} introduces the model. Section \ref{Sec_result} presents our existence results. Section \ref{Sec_app} provides an application of the balancedness condition. Proofs are relegated to the Appendix.

\section{Model\label{Sec_M}}

There is a set $F$ of firms and a set $W$ of workers. Let ${\o}$ be the null firm, representing not being matched with any firm. Each worker $w\in W$ has a strict, transitive, and complete preference $\succ_w$ over $\widetilde{F}:=F\cup\{{\o}\}$. For any $f, f'\in \widetilde{F}$, we write $f\succ_w f'$ when $w$ prefers $f$ to $f'$ according to $\succ_w$. We write $f\succeq_w f'$ if either $f\succ_w f'$ or $f=f'$. Let $\succ_W$ denote the preference profile of all workers. Each firm $f\in F$ has a strict, transitive and complete preference $\succ_f$ over $2^W$. For any $S,S'\subseteq W$, we write $S\succ_f S'$ when $f$ prefers $S$ to $S'$ according to $\succ_f$. We write $S\succeq_f S'$ if either $S\succ_f S'$ or $S=S'$. Let $\succ_F$ be the preference profile of all firms. A matching market can be summarized as a tuple $\Gamma=(W,F,\succ_W,\succ_F)$.

Let $Ch_f$ be the choice function of $f$ such that for any $S\subseteq W$, $Ch_f(S)\subseteq S$ and $Ch_f(S)\succeq_f S'$ for any $S'\subseteq S$. By convention, let $Ch_{{\o}}(S)=S$ for all $S\subseteq W$. For any $f\in F$, any $w\in W$, and any nonempty set $S\subseteq W$, we say that $f$ is acceptable to $w$ if $f\succ_w {\o}$; we say that $S$ is \textbf{acceptable} to $f$ if $S=Ch_f(S)$. The definition of an acceptable set for $f$ in this paper is slightly stronger than the usual one that only requires $S\succ_f\emptyset$. For example, if $f$ has preference $\{w\}\succ\{w,w'\}\succ\emptyset$, then $\{w,w'\}$ is not an acceptable set for $f$ under our definition. For each $f\in F$, let $\mathcal{A}_f$ be the collection of $f$'s acceptable sets. Let $\mathcal{A}=\cup_{f\in F}\mathcal{A}_f$ be the collection of all firms' acceptable sets.\footnote{We do not define acceptable sets for the null firm ${\o}$. By firms' acceptable sets we mean acceptable sets for firms from $F$.}

\begin{definition}
\normalfont
A \textbf{matching} $\mu$ is a function from the set $\widetilde{F}\cup W$ into $\widetilde{F}\cup 2^W$ such that for all $f\in \widetilde{F}$ and $w\in W$,
\begin{description}
\item[(\romannumeral1)] $\mu(w)\in \widetilde{F}$;

\item[(\romannumeral2)] $\mu(f)\in 2^W$;

\item[(\romannumeral3)] $\mu(w)=f$ if and only if $w\in\mu(f)$.
\end{description}
\end{definition}

We say that a matching $\mu$ is \textbf{individually rational} if $\mu(w)\succeq_w {\o}$ for all $w\in W$ and $\mu(f)=Ch_f(\mu(f))$ for all $f\in F$. We say that a firm $f$ and a subset of workers $S\subseteq W$ form a \textbf{blocking coalition} that blocks $\mu$ if $f\succeq_w\mu(w)$ for all $w\in S$, and $S\succ_f\mu(f)$. Individual rationalities for workers mean that each matched worker prefer her current employer to being unmatched.  Individual rationalities for firms mean that no firm wish to unilaterally drop any of its employees. When firm $f$ and a set $S$ block $\mu$, the firm may drop some of its original employees and hire some new workers. Thus, $S$ may contain workers that are matched with $f$ in $\mu$, and each worker $w\in S$ weakly prefers $f$ to $\mu(w)$. However, $f$ strictly prefers $S$ to $\mu(f)$ since $S\neq\mu(f)$.

\begin{definition}\label{stability}
\normalfont
A matching $\mu$ is \textbf{stable} if it is individually rational and there is no blocking coalition that blocks $\mu$.\footnote{The no blocking coalition condition implies individual rationalities for firms: $\mu(f)=Ch_f(\mu(f))$ for all $f\in F$. The stable matching defined by Definition \ref{stability} is called the core defined by weak domination in \cite{RS90}.}
\end{definition}

A stable matching always exists when firms have substitutable preferences in the sense of (\ref{sub}); see Chapter 6 of \cite{RS90}. This paper studies two-sided matching in which each firm has a complementary preference defined as follows.

\begin{definition}
\normalfont
Firm $f$ has a \textbf{complementary} preference if for any $S,S'$ with $S\subset S'\subseteq W$, $w\in Ch_f(S)$ implies $w\in Ch_f(S')$.
\end{definition}

\section{Result}\label{Sec_result}

We present our main results in this section. Section \ref{Sec_balance} provides a sufficient condition for the existence of a stable matching when firms have complementary preferences. Section \ref{Sec_uni} clarifies the relation between our condition and H22's total unimodularity condition. Section \ref{Sec_proof} provides a sketch of the proof. Section \ref{Sec_gen} and Section \ref{Sec_add} present two generalizations of our condition.

\subsection{Existence theorem}\label{Sec_balance}

A 0-1 matrix is \textbf{balanced} if it has no square submatrix of odd order with exactly two 1s in each row and column (\citealp{B70}, see also Chapter 21.5 of \citealp{S86}). That is, it has no submatrix of form (possibly after permutation of rows or columns):
\begin{equation}\label{BM}
\begin{pmatrix}
1      & 1      & 0      & 0      & \cdots & 0     \\
0      & 1      & 1      & 0      & \cdots & 0     \\
0      & 0      & 1      & 1      & \cdots & 0     \\
\vdots & \cdots & \cdots & \cdots & \cdots & \vdots\\
0      & \cdots & \cdots & \cdots & \cdots & 0     \\
0      & \cdots & \cdots & \cdots & 1      & 1     \\
1      & 0      & 0      & \cdots & 0      & 1
\end{pmatrix}
\end{equation}
of odd order.

We say that some sets of workers form matrix $M$ if the columns of $M$ are the indicator vectors of these sets. We show that balancedness of firms' acceptable sets implies the existence of a stable matching when firms have complementary preferences.

\begin{theorem}\label{thm_exist1}
\normalfont
There always exists a stable matching if firms have complementary preferences and firms' acceptable sets form a balanced matrix.
\end{theorem}

Balancedness can be tested in polynomial time; see \cite{CCR99}. Our balancedness condition can be represented as a visualized constraint on firms' acceptable sets by applying Berge's original definition of balancedness on hypergraphs (\citealp{B70}, see also \citealp{B89}). A hypergraph is a generalization of a graph in which an edge can join any number of vertices. Let $\mathcal{A}^*=\{A|A\in \mathcal{A}\text{ and }|A|\geq2\}$ be the collection of firms' non-singleton acceptable sets. The \textbf{acceptable-set hypergraph} is a hypergraph $(W,\mathcal{A}^*)$ in which the vertices are the workers from $W$ and the edges are firms' non-singleton acceptable sets. A cycle of length $k\geq 2$ is a cyclic alternating sequence of distinct workers (vertices) and acceptable sets (edges): $(w^1, A^1, w^2, A^2, ..., w^k, A^k, w^1)$, where $w^i,w^{i+1}\in A^i (i\in\{1,2,\ldots,k-1\})$ and $w^k,w^1\in A^k$.\footnote{The sequence $(w,A,w)(w\in A, A\in \mathcal{A})$ is not considered to be a cycle.}

\begin{fact}
\normalfont
Firms' acceptable sets form a balanced matrix if every odd-length cycle in the acceptable-set hypergraph has an acceptable set containing at least three workers of the cycle.
\end{fact}

\begin{example}
\normalfont
The following is a profile of firms' complementary preferences and its acceptable-set hypergraph.
\begin{equation}
\begin{aligned}
&f_1: \{w_1,w_2,w_3\}\succ\emptyset\\
&f_2: \{w_1,w_2\}\succ\{w_1\}\succ\emptyset\\
&f_3: \{w_2,w_3\}\succ\{w_2\}\succ\emptyset\\
&f_4: \{w_3,w_4\}\succ\{w_3\}\succ\emptyset\\
&f_5: \{w_4,w_5,w_6\}\succ\emptyset
\end{aligned}
\end{equation}
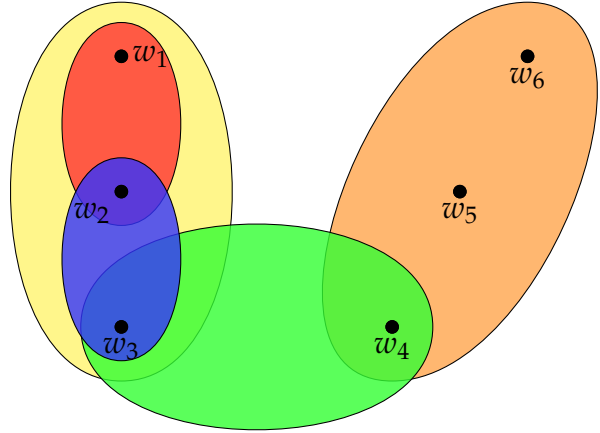
\begin{center}
\begin{tikzpicture}[scale=0.9]
    \node (v1) at (0,4) {};
    \node (v2) at (0,2) {};
    \node (v3) at (0,0) {};
    \node (v4) at (4,0) {};
    \node (v5) at (5,2) {};
    \node (v6) at (6,4) {};

    \begin{scope}[fill opacity=0.8]
    \filldraw[fill=yellow!70] ($(v1)+(0,0.8)$)
        to[out=0,in=0] ($(v3) + (0,-0.8)$)
        to[out=180,in=180] ($(v1) + (0,0.8)$);
    \filldraw[fill=orange!70] ($(v6)+(0,0.8)$)
        to[out=0,in=0] ($(v4) + (0,-0.8)$)
        to[out=180,in=180] ($(v6) + (0,0.8)$);
    \filldraw[fill=red!80] ($(v1)+(0,0.5)$)
        to[out=0,in=0] ($(v2) + (0,-0.5)$)
        to[out=180,in=180] ($(v1) + (0,0.5)$);
    \filldraw[fill=green!80] ($(v3)+(-0.6,0)$)
        to[out=90,in=90] ($(v4) + (0.6,0)$)
        to[out=270,in=270] ($(v3) + (-0.6,0)$);
    \filldraw[fill=blue!80] ($(v2)+(0,0.5)$)
        to[out=0,in=0] ($(v3) + (0,-0.5)$)
        to[out=180,in=180] ($(v2) + (0,0.5)$);
    \end{scope}

    \foreach \v in {1,2,3,4,5,6} {
        \fill (v\v) circle (0.1);
    }

    \fill (v1) circle (0.1) node [right] {$w_1$};
    \fill (v2) circle (0.1) node [below left] {$w_2$};
    \fill (v3) circle (0.1) node [below] {$w_3$};
    \fill (v4) circle (0.1) node [below] {$w_4$};
    \fill (v5) circle (0.1) node [below] {$w_5$};
    \fill (v6) circle (0.1) node [below] {$w_6$};
\end{tikzpicture}
\end{center}
There is only one odd-length cycle $(w_1,\{w_1,w_2\},w_2,\{w_2,w_3\},w_3,\{w_1,w_2,w_3\},w_1)$ in which $\{w_1,w_2,w_3\}$ contains all three workers of the cycle. Therefore, the balancedness condition holds. We then know that this market admits a stable matching for all possible preferences of workers.
\end{example}

\subsection{Relation to total unimodularity of firms' demand type}\label{Sec_uni}

H22 showed that stable matchings exist if firms' demand type is totally unimodular. A matrix is \textbf{totally unimodular} if every square submatrix has determinant 0 or $\pm1$. A totally unimodular matrix is balanced since matrix (\ref{BM}) has determinant 2. A balanced matrix may not be totally unimodular: A balanced matrix is either totally unimodular or has a particular bipartite representation; see \cite{CCR99}. In this section, we discuss the relation between our balancedness condition and H22's total unimodularity condition.

\cite{BK19} proposed the notion of demand type in an exchange economy with continuous monetary transfers and quasi-linear utilities. H22 adopted this notion in a discrete matching market to describe how a firm's choice changes as its available set expands. See Definition 4 of H22. When a firm's available set expands from $\emptyset$ to an acceptable set of the firm, the firm's choice also changes from $\emptyset$ to this acceptable set. Thus, all acceptable sets of a firm belong to the firm's demand type. However, a firm's demand type often contains elements that are not the firm's acceptable sets. For example, suppose $f$ has preference $\{w_1,w_2\}\succ\{w_1\}\succ\emptyset$. $f$'s choice changes from $\{w_1\}$ to $\{w_1,w_2\}$ when its available set expands from $\{w_1\}$ to $\{w_1,w_2\}$. Hence, $(1,1)-(1,0)=(0,1)$ is an element of $f$'s demand type, but $\{w_2\}$ is not acceptable for $f$.

When firms have complementary preferences, our balancedness condition generalizes H22's total unimodularity condition in the following two aspects. (i) Since all acceptable sets of a firm belong to the firm's demand type, total unimodularity of firms' demand type implies total unimodularity (and balancedness) of firms' acceptable sets.\footnote{Any submatrix of a totally unimodular matrix is itself totally unimodular.} However, the converse is not true since a firm's demand type often contains elements that are not the firm's acceptable sets. (ii) Total unimodularity of firms' acceptable sets implies balancedness of firms' acceptable sets. The converse is not true since a balanced matrix may not be totally unimodular.

The following are two illustrative examples. Example \ref{exam_(i)} provides an instance of firms' complementary preferences in which total unimodularity of the firms' demand type fails, but total unimodularity (and balancedness) of the firms' acceptable sets holds. Example \ref{exam_(ii)} provides an instance of firms' complementary preferences in which total unimodularity of the firms' acceptable sets fails, but balancedness of the firms' acceptable sets holds.

\begin{example}\label{exam_(i)}
\normalfont
Consider the following firms' complementary preferences.
\begin{equation}\label{H22}
\begin{aligned}
&f_1: \{w_1,w_2,w_3\}\succ\{w_3\}\succ\emptyset\\
&f_2: \{w_1,w_3\}\succ\emptyset\\
&f_3: \{w_2,w_3\}\succ\emptyset
\end{aligned}
\end{equation}
$f_1$'s choice changes from $\{w_3\}$ to $\{w_1,w_2,w_3\}$ when its available set expands from $\{w_3\}$ to $\{w_1,w_2,w_3\}$. Hence, $(1,1,1)-(0,0,1)=(1,1,0)$ is an element of $f_1$'s demand type, whereas $\{w_1,w_2\}$ is not an acceptable set for $f_1$. Readers can check that firms' acceptable sets form a totally unimodular (and thus balanced) matrix. Since $f_2$'s demand type contains $(0,1,1)$ and $f_3$'s demand type contains $(1,0,1)$, the firms' demand type is not totally unimodular according to the matrix formed by $(1,1,0)$, $(0,1,1)$, and $(1,0,1)$.
\end{example}

\begin{example}\label{exam_(ii)}
\normalfont
The following is a profile of firms' complementary preferences and its acceptable-set hypergraph.
\begin{equation}
\begin{aligned}
&f_1: \{w_1,w_2,w_3,w_4\}\succ\{w_1,w_2\}\succ\emptyset\\
&f_2: \{w_1,w_3\}\succ\{w_3\}\succ\emptyset\\
&f_3: \{w_1,w_4\}\succ\{w_4\}\succ\emptyset
\end{aligned}
\end{equation}
\begin{center}
\begin{tikzpicture}[scale=0.9]
    \node (v1) at (0,3) {};
    \node (v2) at (3,3) {};
    \node (v3) at (0,0) {};
    \node (v4) at (3,0) {};

    \begin{scope}[fill opacity=0.8]
    \filldraw[fill=yellow!70] ($(v1)+(0,1.5)$)
        to[out=0,in=90] ($(v2) + (1.3,0)$)
        to[out=270,in=0] ($(v4) + (0,-1.3)$)
        to[out=180,in=270] ($(v3) + (-1.5,0)$)
        to[out=90,in=180] ($(v1) + (0,1.5)$);
    \filldraw[fill=red!70] ($(v1)+(-0.6,0)$)
        to[out=90,in=180] ($(v2)+(0,0.7)$)
        to[out=0,in=90] ($(v2)+(0.8,0)$)
        to[out=270,in=0] ($(v2)+(0,-0.7)$)
        to[out=180,in=270] ($(v1)+(-0.6,0)$);
    \filldraw[fill=blue!70] ($(v1)+(0,0.7)$)
        to[out=0,in=90] ($(v3)+(0.4,0)$)
        to[out=270,in=0] ($(v3)+(0,-0.5)$)
        to[out=180,in=180] ($(v1)+(0,0.7)$);
    \filldraw[fill=green!80] ($(v1)+(0,0.7)$)
        to[out=340,in=90] ($(v4)+(0.4,0)$)
        to[out=270,in=0] ($(v4)+(0,-0.7)$)
        to[out=180,in=160] ($(v1)+(0,0.7)$);
    \end{scope}

    \foreach \v in {1,2,3,4} {
        \fill (v\v) circle (0.1);
    }

    \fill (v1) circle (0.1) node [right] {$w_1$};
    \fill (v2) circle (0.1) node [below left] {$w_2$};
    \fill (v3) circle (0.1) node [below] {$w_3$};
    \fill (v4) circle (0.1) node [below] {$w_4$};
\end{tikzpicture}
\end{center}
There are several odd-length cycles, such as the cycle
\begin{equation*}
(w_3,\{w_1,w_3\},w_1,\{w_1,w_4\},w_4,\{w_1,w_2,w_3,w_4\},w_3).
\end{equation*}
However, every odd-length cycle must contain $\{w_1,w_2,w_3,w_4\}$, which has all workers of the cycle. This preference profile satisfies the balancedness condition, but firms' acceptable sets does not form a totally unimodular matrix: The submatrix formed by the four acceptable sets in colour has determinant $2$.
\end{example}

\subsection{Sketch of the proof}\label{Sec_proof}

H22 studied two-sided matching by constructing a fictitious market $\widehat{\Gamma}$ where each worker is ``divisible''. Each stable integral matching in $\widehat{\Gamma}$ corresponds to a stable matching in the original market $\Gamma$. The fictitious market $\widehat{\Gamma}$ is an instance of the continuum market of \cite{CKK19}, and a stable matching (possibly fractional) always exists in $\widehat{\Gamma}$ according to their existence theorem. H22 investigated when we can always transform a stable fractional matching into a stable integral one.

We prove our existence theorem by applying H22's method. There are two critical adjustments of H22's method in our proof.

(i) Given a market $\Gamma$ with firms' complementary preferences, we construct a new market $\overline{\Gamma}$ where each firm that has more than one acceptable sets is decomposed into several firms that each has only one acceptable set. We show that a stable matching in $\overline{\Gamma}$ corresponds to a stable matching in $\Gamma$.

(ii) H22 applied a result of \cite{HK56} on unimodular matrices (see also Theorem 21.5 of \citealp{S86}) to the continuum market induced from the original market $\Gamma$. We instead apply a result of \cite{FHO74} on balanced matrices (see also Theorem 21.7 of \citealp{S86}) to the continuum market induced from $\overline{\Gamma}$.

We illustrate (i) by the following example and provide a formal discussion in Section \ref{Sec_decomp1} of the Appendix. We explain (ii) in Section \ref{proof_thm1} of the Appendix.

\begin{example}\label{exam_decomp}
\normalfont
Consider a market $\Gamma$ with the following preference profile.
\begin{equation}\label{original}
\begin{aligned}
&f_1: \{w_1,w_2,w_3\}\succ\{w_1\}\succ\{w_2,w_3\}\succ\emptyset \quad &w_1: &\quad f_1\succ{\o}\\
&f_2: \{w_2,w_3,w_4\}\succ\emptyset \quad &w_2: &\quad f_1\succ f_2\succ{\o}\\
&\quad &w_3: &\quad f_2\succ f_1\succ{\o}\\
&\quad &w_4: &\quad f_2\succ{\o}
\end{aligned}
\end{equation}
We decompose $f_1$ into three firms $f^1_1$, $f^2_1$, and $f^3_1$: $f^1_1$ has preference $\{w_1,w_2,w_3\}\succ\emptyset$; $f^2_1$ has preference $\{w_1\}\succ\emptyset$; and $f^3_1$ has preference $\{w_2,w_3\}\succ\emptyset$. We also decompose $f_1$ in each worker's preference list at $f_1$'s original position in the list. At each position, the order of the three firms is the same as the order of their acceptable sets in $f_1$'s preference list. The agents have the following preferences in the new market $\overline{\Gamma}$.
\begin{equation}\label{newmarket}
\begin{aligned}
&f^1_1: \{w_1,w_2,w_3\}\succ\emptyset \quad &w_1:&\quad f^1_1\succ f^2_1\succ f^3_1\succ{\o}\\
&f^2_1: \{w_1\}\succ\emptyset \quad &w_2:&\quad f^1_1\succ f^2_1\succ f^3_1\succ f_2\succ {\o}\\
&f^3_1: \{w_2,w_3\}\succ\emptyset \quad &w_3:&\quad f_2\succ f^1_1\succ f^2_1\succ f^3_1\succ{\o}\\
&f_2: \{w_2,w_3,w_4\}\succ\emptyset \quad &w_4: &\quad f_2\succ{\o}
\end{aligned}
\end{equation}
A critical observation is that, in any stable matching of
$\overline{\Gamma}$, only one firm among those decomposed from the same firm can be matched with workers. For instance, only one firm from $\{f^1_1,f^2_1,f^3_1\}$ can be matched with workers in a stable matching of (\ref{newmarket}). Otherwise, suppose $f^2_1$ and $f^3_1$ are both matched with their acceptable sets. According to $f_1$'s complementary preference in (\ref{original}), there is a firm in (\ref{newmarket}) (i.e., $f^1_1$) whose unique acceptable set is the union of the acceptable sets of $f^2_1$ and $f^3_1$.\footnote{This property holds when firms' preferences satisfy a milder additivity condition; see Definition \ref{def_additive}.\label{foot_add}} Because $f^1_1$ must be matched with $\emptyset$ in this case and all workers prefer $f^1_1$ to $f^2_1$ and $f^3_1$, according to our decomposition, $f^1_1$ would form a blocking coalition with workers that are matched with $f^2_1$ and $f^3_1$.

We then find that every stable matching $\overline{\mu}$ in $\overline{\Gamma}$ corresponds to a stable matching $\mu$ in $\Gamma$ as follows: If firm $\overline{f}$ in $\overline{\Gamma}$ is matched with a nonempty set $S$ in $\overline{\mu}$, and $\overline{f}$ is decomposed from $f$ in $\Gamma$, we match $f$ with $S$ in $\mu$. For example, the following is a stable matching in (\ref{newmarket}).
\begin{equation*}\label{matching1}
\left(
             \begin{aligned}
             f^1_1  \qquad & f^2_1  \qquad & f^3_1 \qquad & f_2 \qquad & {\o} \\
             w_1,w_2,w_3 \quad & & & & w_4
             \end{aligned}
\right)
\end{equation*}
This matching corresponds to the following stable matching in the original market (\ref{original}).
\begin{equation*}\label{matching1o}
\left(
             \begin{aligned}
             f_1 \qquad \qquad & f_2 &\qquad \quad {\o} \\
             w_1,w_2,w_3 \qquad  & & w_4
             \end{aligned}
\right)
\end{equation*}
\end{example}
This observation and H22's existence theorem immediately implies that total unimodularity of firms' acceptable sets guarantees the existence of a stable matching when firms have complementary preferences. Total unimodularity of firms' acceptable sets can be further generalized into balancedness of firms' acceptable sets by step (ii), which is explained in Section \ref{proof_thm1} of the Appendix.

\subsection{Primitive acceptable sets}\label{Sec_gen}
Consider preference profile (\ref{exam_in1}).
If we replace $f_1$'s preference with $\{w_1,w_2\}\succ\{w_2\}\succ\emptyset$, we will find that a stable matching does not exist either. However, it will no longer be the case if $f_1$ has the following preference.

\begin{equation}\label{pri}
f_1: \{w_1,w_2\}\succ\{w_1\}\succ\{w_2\}\succ\emptyset
\end{equation}

Readers can check that the new firms' preference profile admits a stable matching for all possible preferences of workers. We can explain this distinction by examining the structure of complementarities in $f_1$'s preference. When $f_1$ has preference (\ref{pri}), $w_1$ and $w_2$ are no longer complements for $f_1$, and thus we can view $f_1$ as two separate firms: one only accepts $w_1$ and the other only accepts $w_2$. Balancedness recovers after we decompose $f_1$ into two firms. This example motivates us to generalize Theorem \ref{thm_exist1} by examining the structure of complementarities inside each firm.

In a matching market $\Gamma$, we describe the structure of complementarities for each firm $f\in F$ by a \textbf{complementarity graph} $G_f=(W_f, E)$ with vertex set $W_f=\{w\in W|w\in A \text{ for some } A\in \mathcal{A}_f\}$. $W_f$ is the set of $f$'s potential employees: Each worker from $W_f$ belongs to some acceptable set for $f$. There is an undirected edge between $w,w'\in W_f$ if they are complements for $f$: For any $w,w'\in W_f$, $(w,w')\in E$ if (i) there exists $S\subseteq W$ such that $w\notin Ch_f(S)$ and $w\in Ch_f(S\cup\{w'\})$; or (ii) there exists $S'\subseteq W$ such that $w'\notin Ch_f(S')$ and $w'\in Ch_f(S'\cup\{w\})$.
\begin{definition}
\normalfont
A nonempty set of workers $S\subseteq W$ is a \textbf{primitive acceptable set} for firm $f$ if $Ch_f(S)=S$ and workers from $S$ are connected in $G_f$.
\end{definition}

A maximal connected subgraph of $G_f$ is called a component of $G_f$. If firm $f$'s complementarity graph $G_f$ has $N$ components, we can view $f$ as a group of $N$ independent subsidiaries. Thus, the notion of primitive acceptable set has the following meaning: $f$'s acceptable set $S$ is a primitive acceptable set for $f$ if workers from $S$ work for only one of $f$'s independent subsidiaries. For example, consider the following firms' complementary preferences.
\begin{equation*}
\begin{aligned}
&f_1: \{w_1,w_2,w_3\}\succ\{w_1,w_2\}\succ\{w_3\}\succ\emptyset\\
&f_2: \{w_1,w_2,w_3\}\succ\{w_1,w_2\}\succ\{w_1\}\succ\{w_2\}\succ\emptyset
\end{aligned}
\end{equation*}
Complementarity graphs $G_{f_1}$ and $G_{f_2}$ are as follows.

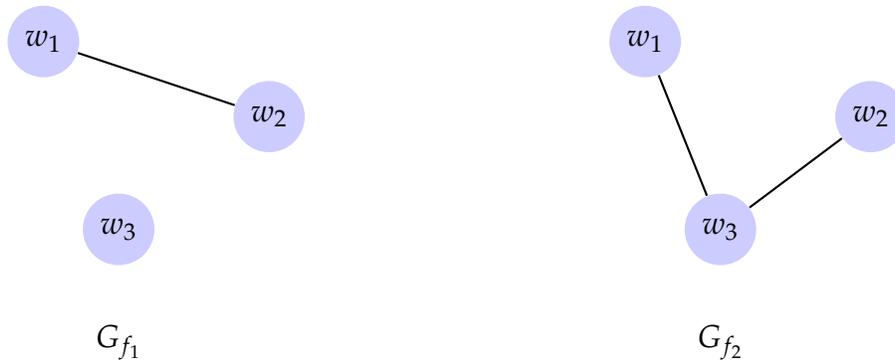
\begin{center}
\tikzstyle{place}=[circle,fill=blue!20]
\begin{tikzpicture}

  \node (n1) at (1,3) [place] {$w_1$};
  \node (n2) at (4,2) [place]{$w_2$};
  \node (n3) at (2,0.5)  [place]{$w_3$};
  \node (n4) at (9,3) [place]{$w_1$};
  \node (n5) at (12,2) [place] {$w_2$};
  \node (n6) at (10,0.5) [place]{$w_3$};
  \node (g1) at (2,-1) {$G_{f_1}$};
  \node (g2) at (10,-1) {$G_{f_2}$};

  \foreach \from/\to in {n1/n2,n4/n6,n5/n6}
    \draw (\from) -- (\to) [thick];

\end{tikzpicture}
\end{center}
$G_{f_1}$ has two components, and thus $f_1$ can be viewed as a group of two independent subsidiaries. Hence, $\{w_1,w_2\}$ and $\{w_3\}$ are $f_1$'s primitive acceptable sets. $G_{f_2}$ has only one component, and thus all acceptable sets for $f_2$ are primitive acceptable sets for $f_2$.\footnote{Note that $\{w_1,w_2\}$ is a primitive acceptable set for $f_2$: Although $w_1$ and $w_2$ are not complements for $f_2$, they are both complements to $w_3$ and thus connected in $G_{f_2}$.}

We have the following theorem that generalizes Theorem \ref{thm_exist1}.

\begin{theorem}\label{thm_exist3}
\normalfont
There always exists a stable matching if firms have complementary preferences and firms' primitive acceptable sets form a balanced matrix.
\end{theorem}

We prove Theorem \ref{thm_exist3} in Section \ref{Sec_decomp2} of the Appendix by decomposing each firm into its independent subsidiaries. We can analogously generalize H22's total unimodularity condition by examining the structure of substitutabilities and complementarities inside each firm. For instance, consider Example 2 of H22. If we view $f_1$ as two separate firms, total unimodularity of firms' demand type also recovers.

\subsection{Additive preferences}\label{Sec_add}

This paper aims to provide a sufficient condition on firms' complementary preferences to guarantee the existence of a stable matching. While we find that the balancedness constraint on firms' complementary preferences guarantees the existence of a stable matching, our proof techniques indicate that stable matchings exist when the balancedness constraint is imposed on a more general class of firms' preferences defined as follows.

\begin{definition}\label{def_additive}
\normalfont
Firm $f$ has an \textbf{additive} preference if for any $S,S'\in \mathcal{A}_f$ with $S\cap S'=\emptyset$, $S\cup S'\in \mathcal{A}_f$.
\end{definition}

A firm has an additive preference if the union of any two disjoint acceptable sets is also acceptable for the firm. Notice that complementary preference is equivalent to the above definition without ``$S\cap S'=\emptyset$''; thus, additive preferences subsume complementary preferences. Firms' additive preferences allow workers to be substitutes. For example, consider the following firms' preferences.

\begin{equation}\label{exam_add}
\begin{aligned}
&f_1: \{w_1,w_2\}\succ\{w_2,w_3\}\succ\emptyset\\
&f_2: \{w_1,w_2,w_3,w_4\}\succ\{w_1,w_2\}\succ\{w_2,w_3\}\succ\{w_3,w_4\}\succ\emptyset
\end{aligned}
\end{equation}
$f_1$ has an additive preference since $f_1$ has no disjoint acceptable sets. $f_2$ has a pair of disjoint acceptable sets: $\{w_1,w_2\}$ and $\{w_3,w_4\}$; $f_2$'s preference is additive since $\{w_1,w_2,w_3,w_4\}$ is also acceptable for $f_2$.
Neither $f_1$ nor $f_2$ has a complementary preference: $f_1$ wants to substitute $w_3$ by $w_1$ when it is matched with $\{w_2,w_3\}$; $f_2$ wants to substitute $w_4$ by $w_2$ when it is matched with $\{w_3,w_4\}$; $f_2$ also wants to substitute $w_3$ by $w_1$ when it is matched with $\{w_2,w_3\}$.

Recall that we only use additivity of firms' complementary preferences in Section \ref{Sec_proof} (and Section \ref{Sec_decomp1}) to prove Theorem \ref{thm_exist1}.\footnote{See footnote \ref{foot_add}.} Therefore, we can generalize Theorem \ref{thm_exist1} into the following theorem.

\begin{theorem}\label{thm_additive}
\normalfont
There always exists a stable matching if firms have additive preferences and firms' acceptable sets form a balanced matrix.\footnote{A similar generalization as Section \ref{Sec_gen} applies to this result; however, we omit this generalization.}
\end{theorem}

For example, the following is the acceptable-set hypergraph for (\ref{exam_add}). Every odd-length cycle contains $\{w_1,w_2,w_3,w_4\}$ that has all vertices of the cycle. Hence, (\ref{exam_add}) admits a stable matching for all possible preferences of workers.

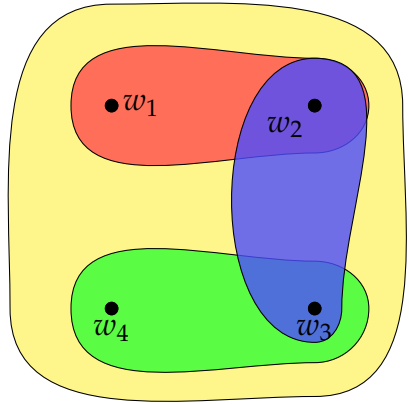
\begin{center}
\begin{tikzpicture}[scale=0.9]
    \node (v1) at (0,3) {};
    \node (v2) at (3,3) {};
    \node (v4) at (0,0) {};
    \node (v3) at (3,0) {};

    \begin{scope}[fill opacity=0.8]
    \filldraw[fill=yellow!70] ($(v1)+(0,1.5)$)
        to[out=0,in=90] ($(v2) + (1.3,0)$)
        to[out=270,in=0] ($(v3) + (0,-1.3)$)
        to[out=180,in=270] ($(v4) + (-1.5,0)$)
        to[out=90,in=180] ($(v1) + (0,1.5)$);
    \filldraw[fill=red!70] ($(v1)+(-0.6,0)$)
        to[out=90,in=180] ($(v2)+(0,0.7)$)
        to[out=0,in=90] ($(v2)+(0.8,0)$)
        to[out=270,in=0] ($(v2)+(0,-0.7)$)
        to[out=180,in=270] ($(v1)+(-0.6,0)$);
    \filldraw[fill=green!80] ($(v4)+(-0.6,0)$)
        to[out=90,in=180] ($(v3)+(0,0.7)$)
        to[out=0,in=90] ($(v3)+(0.8,0)$)
        to[out=270,in=0] ($(v3)+(0,-0.8)$)
        to[out=180,in=270] ($(v4)+(-0.6,0)$);
    \filldraw[fill=blue!70] ($(v2)+(0,0.7)$)
        to[out=0,in=90] ($(v3)+(0.4,0)$)
        to[out=270,in=0] ($(v3)+(0,-0.5)$)
        to[out=180,in=180] ($(v2)+(0,0.7)$);
    \end{scope}

    \foreach \v in {1,2,3,4} {
        \fill (v\v) circle (0.1);
    }

    \fill (v1) circle (0.1) node [right] {$w_1$};
    \fill (v2) circle (0.1) node [below left] {$w_2$};
    \fill (v3) circle (0.1) node [below] {$w_3$};
    \fill (v4) circle (0.1) node [below] {$w_4$};
\end{tikzpicture}
\end{center}

\section{Application\label{Sec_app}}

This section presents an application of the balancedness condition. We follow H22 to describe the structure of firms' preferences as a ``technology tree". Each vertex of the technology tree represents a technology that requires a set of workers to implement. Each edge of the tree is an upgrade from one technology to another that demands more workers. A worker is called a specialist by H22 if the worker engages in only one upgrade. H22 showed that firms' demand type is totally unimodular if firms' acceptable sets are from a technology tree where each worker is a specialist. This section shows that workers can be more versatile in a technology tree that induces a balanced preference profile: Firms' primitive acceptable sets form a balanced matrix if they are from a technology tree where each worker engages in only ``a neighbour of upgrades''. The following is an illustrative example.

\begin{example}\label{exam_app}
\normalfont
A technology tree is depicted as follows. Each vertex from $\{v_0,v_1,v_2,v_3,v_4\}$ represents a technology that requires the set of workers on the right to implement.

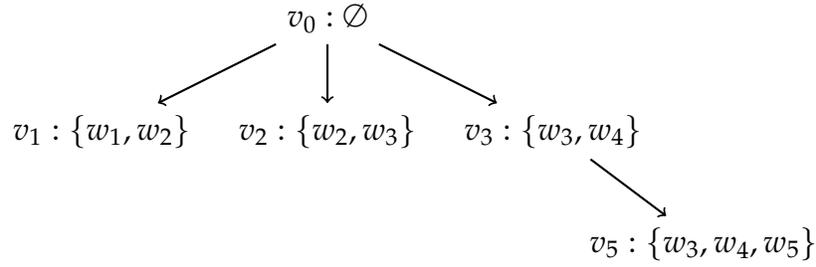
\begin{center}
\begin{tikzpicture}[thick,->,]
\tikzstyle{level 1}=[sibling distance=30mm]
\tikzstyle{level 2}=[sibling distance=20mm]

\node {$v_0:\emptyset$}
	child {node {$v_1:\{w_1,w_2\}$}}
   	child {node {$v_2:\{w_2,w_3\}$}}
    child { node {$v_3:\{w_3,w_4\}$}
		child [missing] {}
		child [missing] {}
		child {node {$v_5:\{w_3,w_4,w_5\}$}
			child [missing] {}
		}
	};
	\end{tikzpicture}
\end{center}

The root $v_0$ represents no technology and requires no worker. Each directed edge is an upgrade from one technology to another that requires more workers. If $e=vv'$ is an edge from vertex $v$ to vertex $v'$, where $w$ is not demanded by $v$ but demanded by $v'$, we say that $w$ engages in upgrade $e$ or $vv'$. For example, $w_2$ engages in upgrades $v_0v_1$ and $v_0v_2$ of the above technology tree. H22 imposed the following constraint on the technology tree:

\begin{equation}\label{specialist}
\text{\emph{Each worker is a specialist that engages in only one upgrade.}}
\end{equation}

\cite{H21b} studied generalizations of H22's market structure. In this section, we present another generalization of (\ref{specialist}). A set of edges is called a neighbour of upgrades if these edges are from the same vertex and not separated by any other upgrade.\footnote{This notion implies that the technology tree is an ordered tree. See the formal definitions below.} For instance, each of $\{v_0v_1\}$, $\{v_0v_2,v_0v_3\}$, and $\{v_0v_1,v_0v_2,v_0v_3\}$ is a neighbour of upgrades, whereas $\{v_0v_1,v_0v_3\}$ is not a neighbour of upgrades since $v_0v_1$ and $v_0v_3$ are separated by $v_0v_2$. We say that $w$ engages in a set $E'$ of upgrades if she engages in every upgrade from $E'$. We relax condition (\ref{specialist}) into the following one.

\begin{equation}\label{neighbour}
\text{\emph{Each worker is a specialist that engages in only a neighbour of upgrades.}}
\end{equation}

The above technology tree satisfies (\ref{neighbour}): Each worker from $\{w_1,w_4,w_5\}$ engages in only one upgrade; $w_2$ and $w_3$ both engage in a neighbour of upgrades. This section shows that if a technology tree satisfies (\ref{neighbour}), worker sets from this tree form a balanced matrix. Consequently, we know that stable matchings exist when firms have complementary preferences and firms' primitive acceptable sets are from a technology tree satisfying (\ref{neighbour}). According to the meaning of primitive acceptable sets discussed in Section \ref{Sec_gen}, we are assuming that each independent subsidiary of each firm wants to implement a technology from a common technology tree that satisfies (\ref{neighbour}). For instance, consider the following firms' complementary preferences.

\begin{equation}\label{pre_app}
\begin{aligned}
&f_1: \{w_1,w_2,w_3,w_4\}\succ \{w_1,w_2\}\succ\{w_3,w_4\}\succ\emptyset\\
&f_2: \{w_3,w_4,w_5\}\succ \{w_3,w_4\}\succ \emptyset
\end{aligned}
\end{equation}
\end{example}
$f_1$ has two primitive acceptable sets: $\{w_1,w_2\}$ and $\{w_3,w_4\}$; $f_2$ has two primitive acceptable sets: $\{w_3,w_4\}$ and $\{w_3,w_4,w_5\}$. All firms' primitive acceptable sets come from the above technology tree, and thus (\ref{pre_app}) admits a stable matching for all possible preferences of workers. To see the role of (\ref{neighbour}), consider the following technology tree.
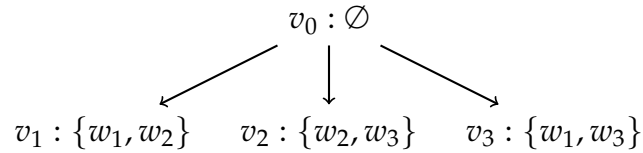
\begin{center}
\begin{tikzpicture}[thick,->]
	\node {$v_0:\emptyset$}
	child {node {$v_1:\{w_1,w_2\}$}}
child [missing] {}
   	child { node {$v_2:\{w_2,w_3\}$}}
   child [missing] {}
    child { node {$v_3:\{w_1,w_3\}$}};
\end{tikzpicture}
\end{center}
This tree violates (\ref{neighbour}): Although $w_2$ and $w_3$ both engage in a neighbour of upgrades, $w_1$ engages in both $v_0v_1$ and $v_0v_3$ that are not neighbours.\footnote{There is no permutation of the three edges that recovers (\ref{neighbour}) either.} This tree induces market (\ref{exam_in1}), which does not admit a stable matching.\footnote{(\ref{exam_in1}) may be induced from other technology trees. Readers can check that none of them satisfies (\ref{neighbour}).}

Formally, a \textbf{technology tree} $T=(V,E,>_V,W)$ is a directed ordered tree $(V,E,>_V)$ defined on a set of workers $W$.\footnote{An ordered tree is a tree where the children of every vertex are ordered, that is, there is a first child, second child, third child, etc.} $V=\{v_0,v_1,\ldots,v_l\}$ is a set of vertices with $v_0$ as the root. Each vertex $v\in V$ represents a \textbf{technology} and requires a subset of workers $W^v\subseteq W$ to implement. The root $v_0$ represents no technology and requires no worker: $W^{v_0}=\emptyset$. $E$ is a set of directed edges, all of which point away from the root. For each edge $e\in E$ from vertex $v$ to vertex $v'$, $W^v\subset W^{v'}$, and we let $W^e=W^{v'}\setminus W^{v}$. We say that $w$ engages in upgrade $e$ if $w\in W^e$. For each $v\in V$, let $E_{v\rightarrow}$ denote the set of edges that point from $v$. For each vertex $v\in V$, there is a complete and transitive order $>_v$ over $E_{v\rightarrow}$. In this section, we draw technology trees such that $e$ is on the left of $e'$ for any $e,e'\in E_{v\rightarrow}$ if $e>_ve'$. A set of edges $E'\subseteq E_{v\rightarrow}$ for some $v\in V$ is called \textbf{a neighbour of upgrades} if there exist no $e,e',e''\in E_{v\rightarrow}$ such that $e,e''\in E'$, $e'\notin E'$, and $e>_ve'>_ve''$.

We find that worker sets on a technology tree form a balanced matrix when each worker engages in only a neighbour of upgrades.

\begin{theorem}\label{thm_neighbour}
\normalfont
In a technology tree where each worker engages in only a neighbour of upgrades, elements from $\{W^v|v\in V\}$ form a balanced matrix.\footnote{More specifically, elements from $\{W^v|v\in V\}$ form a special balanced matrix called \emph{totally balanced}; see Section \ref{proof_thm4}. A 0-1 matrix is totally balanced if it does not contain as a submatrix the incidence matrix of any cycle of length at least 3.}
\end{theorem}

Given a matching market $\Gamma$, let $\mathcal{A}^p_f$ be the collection of $f$'s primitive acceptable sets for each $f\in F$.  We say that $f$'s primitive acceptable sets are from a technology tree $T=(V,E,>_V,W)$ if $S\in \mathcal{A}^p_f$ implies $S=W^v$ for some $v\in V$. According to Theorem \ref{thm_exist3} and Theorem \ref{thm_neighbour}, we know that stable matchings exist when firms have complementary preferences and firms' primitive acceptable sets are from a technology tree satisfying (\ref{neighbour}).\footnote{We can generalize this statement slightly by only require firms' non-singleton primitive acceptable sets to be from a proper technology tree. For example, if $f_2$ has preference $\{w_3,w_4,w_5\}\succ \{w_3\}\succ \emptyset$ in (\ref{pre_app}), all firms' non-singleton primitive acceptable sets are from the technology tree in Example \ref{exam_app}, and thus firms' primitive acceptable sets form a balanced matrix.} Similarly, according to Theorem \ref{thm_additive} and Theorem \ref{thm_neighbour}, we also know that stable matchings exist when firms have additive preferences and firms' acceptable sets are from a technology tree that satisfies (\ref{neighbour}). Technology trees satisfying (\ref{neighbour}) can also induce complementary preference profiles that fail H22's total unimodularity condition. For example, consider the following technology tree, which satisfies (\ref{neighbour}).
\begin{center}
\begin{tikzpicture}[thick,->]
\tikzstyle{level 1}=[sibling distance=20mm]
\tikzstyle{level 2}=[sibling distance=35mm]
	\node {$v_0:\emptyset$}
	child [missing] {}
   	child { node {$v_1:\{w_3\}$}
        child {node {$v_2:\{w_1,w_3\}$}}
		child {node {$v_3:\{w_1,w_2,w_3\}$}}
		child {node {$v_4:\{w_2,w_3\}$}}
		}
   child [missing] {};	
	\end{tikzpicture}
\end{center}

Firms' complementary preferences in (\ref{H22}) can be induced from this technology tree. (\ref{H22}) satisfies the balancedness condition but fails H22's total unimodularity condition as we explain in Example \ref{exam_(i)}.

\bigskip

\emph{(The following remark does not belong to the topic of complementary preferences and will be moved into another working paper.)}

\begin{remark}
\normalfont

We obtain another existence result by introducing a notion of firm-worker hypergraphs. A \textbf{firm-worker hypergraph} is a hypergraph $(F\cup W,\mathcal{Y})$ in which each vertex is either a firm or a worker, and each edge contains one firm and some workers that form one of the firm's acceptable sets:
$\mathcal{Y}=\{\{f\}\cup S|f\in F \text{ and } S\in \mathcal{A}_f\}$. We call a hypergraph balanced if every odd-length cycle has an edge containing at least three vertices of the cycle.

\begin{theorem}\label{thm_balanced}
\normalfont
There always exists a stable matching if the firm-worker hypergraph is balanced.
\end{theorem}

If the firm-worker hypergraph is balanced, matrix $B$ in Section 3.3 and Section 5.2 of H22 is balanced. The existence of a stable matching follows immediately from H22's arguments. Consider the following example from \cite{CKK19} and the firm-worker hypergraph induced from the firms' preferences.
\begin{equation}\label{exam_ckk}
\begin{aligned}
&f_1: \{w_1,w_2\}\succ\emptyset \qquad\qquad\qquad\qquad &w_1: &\quad f_1\succ f_2\\
&f_2: \{w_1\}\succ\{w_2\}\succ\emptyset \qquad\qquad\qquad\qquad &w_2: &\quad f_2\succ f_1
\end{aligned}
\end{equation}
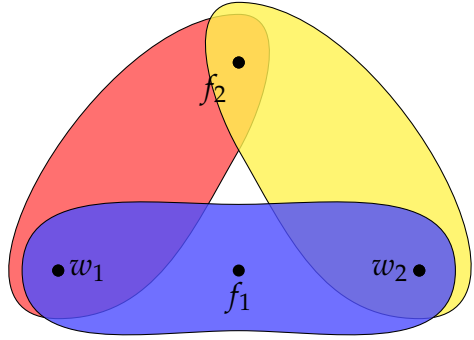
\begin{center}
\begin{tikzpicture}[scale=0.8]
    \node (v1) at (-1,0) {};
    \node (v2) at (2,3.46) {};
    \node (v3) at (5,0) {};
    \node (v4) at (2,2) {};
    \node (v5) at (2,-1) {};
    \node (v6) at (2,0) {};

    \begin{scope}[fill opacity=0.8]
    \filldraw[fill=red!70] ($(v2)+(0,0.8)$)
        to[out=0,in=60] ($(v4)$)
        to[out=240,in=0] ($(v1) + (0,-0.8)$)
        to[out=180,in=180] ($(v2) + (0,0.8)$);
    \filldraw[fill=yellow!80] ($(v2)+(0,1)$)
        to[out=180,in=120] ($(v4)$)
        to[out=300,in=180] ($(v3) + (0,-0.8)$)
        to[out=0,in=0] ($(v2) + (0,1)$);
    \filldraw[fill=blue!70] ($(v1)+(-0.6,0)$)
        to[out=90,in=180] ($(v4)+ (0,-0.9)$)
        to[out=0,in=90] ($(v3) + (0.6,0)$)
        to[out=270,in=0] ($(v5)$)
        to[out=180,in=270] ($(v1) + (-0.6,0)$);
    \end{scope}

    \foreach \v in {1,2,3} {
        \fill (v\v) circle (0.1);
    }

    \fill (v1) circle (0.1) node [right] {$w_1$};
    \fill (v2) circle (0.1) node [below left] {$f_2$};
    \fill (v3) circle (0.1) node [left] {$w_2$};
    \fill (v6) circle (0.1) node [below] {$f_1$};

\end{tikzpicture}
\end{center}
\medskip

Market (\ref{exam_ckk}) has no stable matching.\footnote{In this market, $f_1$ should hire both workers or neither of them in any stable matching. In the former case, $f_2$ would form a blocking coalition with $w_2$, who prefers $f_2$ to $f_1$. In the latter case, $f_2$ would be matched with $w_1$, leaving $w_2$ unmatched. Then, $f_1$ would form a blocking coalition with both $w_1$ and $w_2$.} The nonexistence of a stable matching in this market is attributed to the above odd-length cycle. In a subsequent work, we study balanced firm-worker hypergraphs and show that balancedness of the firm-worker hypergraph is independent of total unimodularity of firms' demand type.
\end{remark}

\section{Appendix}

\subsection{Firm decomposition I}\label{Sec_decomp1}

Given a matching market $\Gamma=(W,F,\succ_W,\succ_F)$ in which firms have complementary preferences, we construct a new market $\overline{\Gamma}=(W,\overline{F},\overline{\succ}_{W},\overline{\succ}_{\overline{F}})$ where each firm $f\in F$ in $\Gamma$ is decomposed into a set $H_f$ of $|\mathcal{A}_f|$ firms: $H_f=\{f^1,\ldots,f^{|\mathcal{A}_f|}\}$.\footnote{If $f$ has only one acceptable set in $\Gamma$, then $H_f=\{f^1\}$ contains only one firm. This exposition is different from Example \ref{exam_decomp}.} For each $f\in F$, each $\overline{f}\in H_f$ has only one acceptable set such that $\cup_{\overline{f}\in H_f}\mathcal{A}_{\overline{f}}=\mathcal{A}_f$.

Let $A_{\overline{f}}$ denote $\overline{f}$'s unique acceptable set for each $\overline{f}\in \overline{F}$. We modify workers' preference profile $\succ_{W}$ of $\Gamma$ into $\overline{\succ}_{W}$ of $\overline{\Gamma}$ as follows. For any $\overline{f}\in H_f$, any $\overline{f}'\in H_{f'}$ with $f\neq f'$, and any $w\in W$, $\overline{f}\overline{\succ}_w\overline{f}'$ if $f\succ_wf'$. For any $\overline{f}\in H_f$ and any $w\in W$ $\overline{f}\overline{\succ}_w{\o}$ if $f\succ_w{\o}$. For any $\overline{f},\overline{f}'\in H_f$ and any $w\in W$, $\overline{f}\overline{\succ}_w\overline{f}'$ if $A_{\overline{f}}\succ_fA_{\overline{f}'}$. That is, we decompose each firm $f\in F$ into $|\mathcal{A}_f|$ firms at its original location in each worker's preference list, and at each location the $|\mathcal{A}_f|$ firms are in an order the same as the order of their acceptable sets in $f$'s preference list. See Example \ref{exam_decomp} for an illustration. We have the following lemma.

\begin{lemma}\label{lma_decomp1}
\normalfont
If $\overline{\mu}$ is a stable matching in $\overline{\Gamma}$, then $\mu$ is a stable matching in $\Gamma$ where $\mu(f)=\cup _{\overline{f}\in H_f}\overline{\mu}(\overline{f})$ for each $f\in F$, $\mu(w)=f$ if $\overline{\mu}(w)\in H_f$ for each $w\in W$, and $\mu({\o})=\overline{\mu}({\o})$.
\end{lemma}

\begin{proof}
$\mu$ is a matching since $\mu(w)=f$ if and only if $w\in \mu(f)$ for all $w\in W$ and all $f\in \widetilde{F}$. The individual rationalities for workers in $\Gamma$ follow that $\overline{\mu}(w)\succ_w{\o}$ implies $\mu(w)\succ_w{\o}$ for each $w\in W$. Since for each $\overline{f}\in H_f$, $\overline{\mu}(\overline{f})$ in $\overline{\Gamma}$ is either an acceptable set for $f$ in $\Gamma$ or empty, we know that $\mu(f)$ is acceptable for each $f\in F$ due to firms' complementary preferences (or additive preferences defined in Section \ref{Sec_add}). Hence, $\mu$ is also individually rational for firms.

For each $f\in F$, among all firms from $H_f$, there exists only one $\overline{f}\in H_f$ such that $\overline{\mu}(\overline{f})\neq\emptyset$, and thus $\mu(f)=\overline{\mu}(\overline{f})$. Otherwise, suppose there exist $\overline{f},\overline{f}'\in H_f$ such that $\overline{\mu}(\overline{f})\neq\emptyset$ and $\overline{\mu}(\overline{f}')\neq\emptyset$. Since $\overline{\mu}(\overline{f})$ and $\overline{\mu}(\overline{f}')$ are both acceptable sets for $f$ in $\Gamma$ and $\overline{\mu}(\overline{f})\cap\overline{\mu}(\overline{f}')=\emptyset$, according to firms' complementary preferences (or additive preferences defined in Section \ref{Sec_add}), we know that there exists $\overline{f}''\in H_f$ that has a unique acceptable set $A_{\overline{f}''}=\overline{\mu}(f)\cup\overline{\mu}(f')$. $\overline{\mu}(\overline{f})\neq\emptyset$ and $\overline{\mu}(\overline{f}')\neq\emptyset$ imply $\overline{\mu}(\overline{f}'')=\emptyset$. According to our decomposition, all workers in $\overline{\Gamma}$ prefer $\overline{f}''$ to $\overline{f}$ and $\overline{f}'$, and thus $\overline{f}''$ and $\overline{\mu}(\overline{f})\cup\overline{\mu}(\overline{f}')$ block $\overline{\mu}$. A contradiction.

Suppose $f$ and $S\subseteq W$ block $\mu$ in $\Gamma$, then $Ch_f(S)\succeq_fS\succ_f\mu(f)$. Let $\overline{f}$ be the unique firm from $H_f$ such that $\overline{\mu}(\overline{f})\neq\emptyset$ and $\overline{\mu}(\overline{f})=\mu(f)$. Let $\overline{f}'$ be the firm from $H_f$ such that $A_{\overline{f}'}=Ch_f(S)$. We have $Ch_f(S)\overline{\succ}_{\overline{f}'}\overline{\mu}(\overline{f}')=\emptyset$ and $\overline{f}'\overline{\succ}_w\overline{f}$ for each $w\in W$. Since each $w\in S$ weakly prefers $f$ to $\mu(w)$ in $\Gamma$, according to our decomposition, each $w\in S$ in $\overline{\mu}$  is either matched with $\overline{f}$ or some firm that is less preferred than $\overline{f}$ by $w$. Since $\overline{f}'\overline{\succ}_w\overline{f}$ for each $w\in W$, we know that each $w\in S$ weakly prefers $\overline{f}'$ to $\overline{\mu}(w)$ in $\overline{\Gamma}$. Therefore, $\overline{f}'$ and $Ch_f(S)$ block $\overline{\mu}$ in $\overline{\Gamma}$. A contradiction.
\end{proof}

\subsection{Proof of Theorem \ref{thm_exist1}}\label{proof_thm1}

H22 used a result of \cite{HK56} on unimodular matrices. Our Theorem \ref{thm_exist1} follows from H22's method by instead applying a result of \cite{FHO74} on balanced matrices to the continuum market induced from $\overline{\Gamma}$. In the following, we illustrate how to use the result of \cite{FHO74} in H22's method to prove Theorem \ref{thm_exist1}. We do not give a complete proof in this section but only explain the changes we made in H22' method. We refer the reader to Section 3 and Section 5.2 of H22 and this section for a complete proof. However, in the following we introduce the method in an intuitive way to make the discussion self-contained.

We follow H22 to construct a continuum market $\widehat{\Gamma}$ from $\overline{\Gamma}$ by assuming that each worker is divisible and constructing a particular preference for each firm.\footnote{See Section 3.1 of H22 for the formal definitons of firms' preferences and stable (fractional/integral) matchings in the continuum market.} Each $w\in W$ is called a worker type in $\widehat{\Gamma}$. Workers of type-$w$ in $\widehat{\Gamma}$ have the same preference as $w$ in $\overline{\Gamma}$. In our problem, H22's construction yields a Leontief preference for each firm in $\widehat{\Gamma}$: Each firm in $\widehat{\Gamma}$ wants to hire equal amounts of each worker type from its unique acceptable set in $\overline{\Gamma}$. For example, in the continuum market $\widehat{\Gamma}$ induced from (\ref{newmarket}), $f_1^1$ chooses $\min\{x_1,x_2,x_3\}(1,1,1,0)$ from an available worker share $(x_1,x_2,x_3,x_4)$, where $x_i$ is the amount of type-$w_i$ workers available for $f_1^1$. The existence theorem of \cite{CKK19} implies that there always exists a stable matching in $\widehat{\Gamma}$, which may be fractional or integral. For instance, the continuum market $\widehat{\Gamma}$ induced from (\ref{newmarket}) has the following stable fractional matching $M$.

\begin{center}
\begin{tabular}
[c]{c|cccc}
& $w_1$ & $w_2$ & $w_3$ & $w_4$ \\\hline
$f^1_1$ & $0.5$ & $0.5$ & $0.5$ & $0$ \\
$f^2_1$ & $0.5$ & $0$ & $0$ & $0$ \\
$f^3_1$ & $0$ & $0$ & $0$ & $0$ \\
$f_2$ & $0$ & $0.5$ & $0.5$ & $0.5$ \\
${\o}$ & $0$ & $0$ & $0$ & $0.5$ %
\end{tabular}
\end{center}
This matching is stable in the following sense:\footnote{See Definition 5 of H22 for the formal definition of stability in $\widehat{\Gamma}$.} (a) Individual rationality: No worker is matched with an unacceptable firm for her type; no firm wants to unilaterally drop any of its employees. (b) No blocking coalition: $f_1^1$ wants to hire more workers of type-$w_1$, type-$w_2$, and type-$w_3$ at the ratio $1:1:1$, but $f_1^1$ is not attractive to the type-$w_3$ workers matched with $f_2$ since type-$w_3$ workers prefer $f_2$ to $f_1^1$. Similarly, none of $f_1^2$, $f^3_1$, or $f_2$ can draw workers from other firms to get better off.\footnote{$f_1^2$ wants to hire more type-$w_1$ workers, but $f_1^2$ is not attractive to the type-$w_1$ workers matched with $f_1^1$ since all workers prefer $f_1^1$ to $f_1^2$. $f_1^3$ wants to hire more workers of type-$w_2$ and type-$w_3$ at the ratio $1:1$, but $f_1^3$ is not attractive to the type-$w_3$ workers matched with $f_1^1$ and $f_2$ since type-$w_3$ workers prefer $f_1^1$ and $f_2$ to $f_1^3$. $f_2$ wants to hire more workers of type-$w_2$, type-$w_3$, and type-$w_4$ at the ratio $1:1:1$, but $f_2$ is not attractive to the type-$w_2$ workers matched with $f_1^1$ since type-$w_2$ workers prefer $f_1^1$ to $f_2$.}

\textbf{Stable transformations.} Consider to match $f_1^1$ with $(0.6,0.6,0.6,0)$ in the above matching. We find that the above (a) and (b) still hold for all firms. We call such an assignment a stable pseudo-matching, which may not assign quantity 1 of some worker types. The following operations on a stable pseudo-matching preserve stability.

1. Pick a firm $f$ from $F$, suppose a stable pseudo-matching matches $f$ with quantity $x\in(0,1)$ of each worker type from its unique acceptable set in $\overline{\Gamma}$. Then, either of the following transformations preserves stability: (i) match $f$ with $\mathbf{0}$; (ii) match $f$ with quantity $1$ of each worker type from its unique acceptable set in $\overline{\Gamma}$.

2. Pick a worker type $w$, suppose a stable pseudo-matching matches some positive amount of type-$w$ workers with the null firm ${\o}$. The following transformation preserves stability: (iii) match ${\o}$ with quantity $0$ or $1$ of type-$w$ workers, the amounts of other worker types matched with ${\o}$ remain unchanged.\footnote{Type-(i), type-(ii), and type-(iii) transformations correspond to the type-1, type-2, and type-3 stable transformations of H22, respectively.}

For example, consider the following transformations on the above stable fractional matching.
\medskip

$M=$
\begin{tikzpicture}[baseline = (M.west)]
    \tikzset{brace/.style = {decorate, decoration = {brace, amplitude = 5pt}, thick}}
    \matrix(M)
    [
        matrix of math nodes,
        left delimiter = (,
        right delimiter = )
    ]
    {
0.5 & 0.5 & 0.5 & 0 \\
0.5 & 0 & 0 & 0 \\
0 & 0 & 0 & 0 \\
0 & 0.5 & 0.5 & 0.5 \\
0 & 0 & 0 & 0.5\\
    };
\end{tikzpicture}
$\stackrel{(1)}{\Longrightarrow}$
\begin{tikzpicture}[baseline = (M.west)]
    \tikzset{brace/.style = {decorate, decoration = {brace, amplitude = 5pt}, thick}}
    \matrix(M)
    [
        matrix of math nodes,
        left delimiter = (,
        right delimiter = )
    ]
    {
1 & 1 & 1 & 0 \\
0.5 & 0 & 0 & 0 \\
0 & 0 & 0 & 0 \\
0 & 0.5 & 0.5 & 0.5 \\
0 & 0 & 0 & 0.5\\
    };
\end{tikzpicture}
$\stackrel{(2)}{\Longrightarrow}$
\begin{tikzpicture}[baseline = (M.west)]
    \tikzset{brace/.style = {decorate, decoration = {brace, amplitude = 5pt}, thick}}
    \matrix(M)
    [
        matrix of math nodes,
        left delimiter = (,
        right delimiter = )
    ]
    {
1 & 1 & 1 & 0 \\
0 & 0 & 0 & 0 \\
0 & 0 & 0 & 0 \\
0 & 0.5 & 0.5 & 0.5 \\
0 & 0 & 0 & 0.5\\
    };
\end{tikzpicture}

$\stackrel{(3)}{\Longrightarrow}$
\begin{tikzpicture}[baseline = (M.west)]
    \tikzset{brace/.style = {decorate, decoration = {brace, amplitude = 5pt}, thick}}
    \matrix(M)
    [
        matrix of math nodes,
        left delimiter = (,
        right delimiter = )
    ]
    {
1 & 1 & 1 & 0 \\
0 & 0 & 0 & 0 \\
0 & 0 & 0 & 0 \\
0 & 0 & 0 & 0 \\
0 & 0 & 0 & 0.5\\
    };
\end{tikzpicture}
$\stackrel{(4)}{\Longrightarrow}$
\begin{tikzpicture}[baseline = (M.west)]
    \tikzset{brace/.style = {decorate, decoration = {brace, amplitude = 5pt}, thick}}
    \matrix(M)
    [
        matrix of math nodes,
        left delimiter = (,
        right delimiter = )
    ]
    {
1 & 1 & 1 & 0 \\
0 & 0 & 0 & 0 \\
0 & 0 & 0 & 0 \\
0 & 0 & 0 & 0 \\
0 & 0 & 0 & 1\\
    };
\end{tikzpicture}$=M'$

\medskip
Transformation (1) is a type-(ii) stable transformation. This transformation preserves stability because other firms can not draw workers from $f_1^1$ no matter when $f_1^1$ is matched with $(0.5,0.5,0.5,0)$ or $(1,1,1,0)$. Transformation (2) is a type-(i) stable transformation. Since $f_1^2$ cannot draw type-$w_1$ workers from other firms when matched with $(0.5,0,0,0)$, $f_1^2$ cannot do such things either when matched with $(0,0,0,0)$. Transformation (3) is also a type-(i) stable transformation. Since $f_2$ cannot draw workers of type-$w_2$, type-$w_3$, and type-$w_4$ simultaneously from other firms when matched with $(0,0.5,0.5,0.5)$, $f_2$ cannot do such things either when matched with $(0,0,0,0)$. Transformation (4) is a type-(iii) stable transformation. This transformation preserves stability because changes of the quantities for unmatched workers do not affect stability.\footnote{Note that we can make such changes on a worker type only when a positive amount of this worker type has been matched with ${\o}$.} This process produces a stable integral matching $M'$.

According to H22, every stable integral matching of $\widehat{\Gamma}$ corresponds to a stable matching of $\overline{\Gamma}$.\footnote{See Lemma 2 of H22.} For example, the above stable integral matching $M'$ is also a stable matching in $\overline{\Gamma}$. Then, according to Lemma \ref{lma_decomp1}, the existence of a stable integral matching in $\widehat{\Gamma}$ implies the existence of a stable matching in the original market $\Gamma$. Therefore, we turn to study when stable integral matchings exist in $\widehat{\Gamma}$. H22 showed the following observation: We can always obtain a stable integral matching through stable transformations on stable fractional matching $M$ when there is a nonnegative integral solution to the following system of linear equations.
\medskip

\begin{tikzpicture}[baseline = (M.west)]
    \tikzset{brace/.style = {decorate, decoration = {brace, amplitude = 5pt}, thick}}
    \matrix(M)
    [
        matrix of math nodes,
        left delimiter = (,
        right delimiter = )
    ]
    {
        1 & 1 & 0 & 0 & 0 & 0 & 0\\
        0 & 0 & 1 & 1 & 0 & 0 & 0\\
        0 & 0 & 0 & 0 & 1 & 1 & 0\\
        1 & 0 & 1 & 0 & 0 & 0 & 0\\
        1 & 0 & 0 & 0 & 1 & 0 & 0\\
        1 & 0 & 0 & 0 & 1 & 0 & 0\\
        0 & 0 & 0 & 0 & 1 & 0 & 1\\
            };
     \draw[decorate,decoration={brace,mirror,amplitude=3mm},thick]
        ($(M-4-1.north east) + (-1, 0)$)
        -- node[left = 18pt]{$B^*$}
        ($(M-7-1.south east) + (-1, 0)$);
\end{tikzpicture}
\begin{tikzpicture}[baseline = (M.west)]
    \tikzset{brace/.style = {decorate, decoration = {brace, amplitude = 5pt}, thick}}
        \matrix(M)
    [
        matrix of math nodes,
        left delimiter = (,
        right delimiter = )
    ]
    {
       z_1\\
z_2\\
z_3\\
z_4\\
z_5\\
z_6\\
z_7\\
    };
\end{tikzpicture}$=$
\begin{tikzpicture}[baseline = (M.west)]
    \tikzset{brace/.style = {decorate, decoration = {brace, amplitude = 5pt}, thick}}
        \matrix(M)
    [
        matrix of math nodes,
        left delimiter = (,
        right delimiter = )
    ]
    {
1\\
1\\
1\\
1\\
1\\
1\\
1\\
    };
\end{tikzpicture}

\medskip
Let $B$ denote the $7\times 7$ matrix on the left. The first to third rows of $B$ are constraints for preserving stability; the fourth to seventh rows of $B$ are constraints for assigning quantity 1 of each worker type. Let $B^*$ denote the $4\times 7$ submatrix that includes the fourth to seventh rows of $B$. Let $B_i$ and $B^*_i$ be the $i$-th columns of $B$ and $B^*$, respectively. $B^*_1$ and $B^*_2$ correspond to $\{w_1,w_2,w_3\}$ and $\emptyset$ in $f^1_1$'s preference list, respectively. $B^*_3$ and $B^*_4$ correspond to $\{w_1\}$ and $\emptyset$ in $f^2_1$'s preference list, respectively. $B^*_5$ and $B^*_6$ correspond to $\{w_2,w_3,w_4\}$ and $\emptyset$ in $f_2$'s preference list, respectively. $B^*_7$ corresponds to the type-$w_4$ workers matched with firm ${\o}$. $\mathbf{z}=(0.5,0.5,0.5,0.5,0.5,0.5,0.5)$ is a solution to this system, which refers to matching $M$ in the continuum market as follows.

\begin{equation}\label{z-matching}
M=\left(
\begin{aligned}    f^1_1  \qquad\qquad\qquad & \qquad\qquad\qquad f^2_1 \qquad \\
             z_1(1,1,1,0)+z_2(0,0,0,0) \quad & \quad z_3(1,0,0,0)+z_4(0,0,0,0)
    \end{aligned}
    \right.
\end{equation}

\begin{equation*}
      \left.  \begin{aligned}     f^3_1 \qquad\qquad\quad  & \qquad f_2 \qquad\quad  & \quad {\o} \qquad \\
             (0,0,0,0)\qquad\qquad & z_5(0,1,1,1)+z_6(0,0,0,0)\quad & \quad z_7(0,0,0,1)  \end{aligned}
             \right)
\end{equation*}

The above formula with any integral $\mathbf{z}\in \{0,1\}^7$ that satisfies $z_1+z_2=1$, $z_3+z_4=1$, and $z_5+z_6=1$ (guaranteed by the first to third rows of $B$) can be obtained from $M$ via the stable transformations. For example, transformation (1) means that we let $z_1=1$ and $z_2=0$.\footnote{Similarly, transformation (2) means that we let $z_3=0$ and $z_4=1$; transformation (3) means that we let $z_5=0$ and $z_6=1$; transformation (4) means that we let $z_7=1$.} Letting $z_1=0$ and $z_2=1$ corresponds to a type-(i) transformation and thus also preserves stability. Then, any $\mathbf{z}\in \{0,1\}^7$ that satisfies $B\mathbf{z}=\mathbf{1}$ corresponds to a stable integral matching since $B^*\mathbf{z}=\mathbf{1}$ requires that the workers of each type assigned to firms is of quantity 1. For instance, $\mathbf{z}'=(1,0,0,1,0,1,1)$ is a solution to the system, and we obtain stable integral matching $M'$ by plugging $\mathbf{z}'$ into (\ref{z-matching}). Hence, we can always obtain a stable integral matching via stable transformations if there is an integral point on the polytope $\{\mathbf{z}\mid B\mathbf{z}=\mathbf{1}, \mathbf{z}\geq0\}$.

Since we construct the polytope $\{\mathbf{z}\mid B\mathbf{z}=\mathbf{1}, \mathbf{z}\geq0\}$ from $M$ that corresponds to a point in the polytope, the polytope is nonempty. According to \cite{FHO74}, if $B$ is balanced, all vertices of the polytope $\{\mathbf{z}\mid B\mathbf{z}=\mathbf{1}, \mathbf{z}\geq0\}$ are integral (see also Theorem 21.7 of \citealp{S86}).  Hence, we know that balancedness of $B$ implies the existence of an integral vertex on this polytope.

To check whether $B$ is balanced, the columns for firm ${\o}$ (i.e., $B_7$ in this example) can be removed since each of these columns contain only one 1. The columns for $\emptyset$ (i.e., $B_2$, $B_4$, and $B_6$ in this example) can also be removed for the same reason. We then notice that the first to third rows of the remaining matrix can also be removed since each of these rows contains only one 1. Therefore, we find that it suffices to check the balancedness of the matrix formed by the remaining columns that are firms' acceptable sets (i.e., $B^*_1$, $B^*_3$, and $B^*_5$ in this example).

Recall that the existence theorem of \cite{CKK19} implies the existence of a stable matching in market $\widehat{\Gamma}$, which may be fractional or integral. Pick such a stable matching $M''$. If $M''$ is integral, we are done. If $M''$ is fractional, the above discussion shows that we can always obtain a stable integral matching from $M''$ via the stable transformations when firms' acceptable sets form a balanced matrix. Therefore, balancedness of firms' acceptable sets implies the existence of a stable integral matching in $\widehat{\Gamma}$, which further implies the existence of a stable matching in $\overline{\Gamma}$ and $\Gamma$.

\subsection{Firm decomposition II}\label{Sec_decomp2}

Let $\mathcal{K}_f$ be the set of components of $G_f$. Given a matching market $\Gamma=(W,F,\succ_W,\succ_F)$ where firms have complementary preferences, we construct a new market $\check{\Gamma}=(W,\check{F},\check{\succ}_{W},\check{\succ}_{\check{F}})$ where each firm $f\in F$ in $\Gamma$ is decomposed into a set $Q_f$ of $|\mathcal{K}_f|$ firms: $Q_f=\{f^1,\ldots,f^{|\mathcal{K}_f|}\}$. Each $\check{f}\in Q_f$ corresponds to a component $K_{\check{f}}\in \mathcal{K}_f$. Let $W^{\check{f}}$ be the set of workers that is the vertex set of $K_{\check{f}}$. Each $\check{f}\in Q_f$ has a preference $\check{\succ}_{\check{f}}$ such that $Ch_{\check{f}}(S)=Ch_{f}(S)\cap W^{\check{f}}$ for all $S\subseteq W$. It is easy to see that each $\check{f}\in\check{F}$ has a complementary preference. Each firm is also decomposed at its original position in each worker's preference list: For any $\check{f},\check{f}'\in \check{F}$ with $\check{f}\in Q_f$, $\check{f}'\in Q_{f'}$, and any $w\in W$, $\check{f}\check{\succ}_w\check{f}'$ if $f\succ_wf'$. For any $\check{f}\in \check{F}$ with $\check{f}\in Q_f$, and any $w\in W$, $\check{f}\check{\succ}_w{\o}$ if $f\succ_w{\o}$. For any $f\in F$, workers' preference orders over firms from $Q_f$ can be arbitrary.

The following lemma shows that every stable matching in $\check{\Gamma}$ corresponds to a stable matching in the original market $\Gamma$. Theorem \ref{thm_exist3} follows immediately from Theorem \ref{thm_exist1} and this lemma since acceptable sets for any $\check{f}\in Q_f$ are primitive acceptable sets for $f$.\footnote{If $\check{f}\in Q_f$ and $Ch_{\check{f}}(S)=S$, since $Ch_{\check{f}}(S)=Ch_{f}(S)\cap W^{\check{f}}$, we have $S\subseteq W^{\check{f}}$ and $Ch_{f}(S)=S$. Thus, we know that each acceptable set of $\check{f}\in Q_f$ is a primitive acceptable set of $f$.}

\begin{lemma}
\normalfont
If $\check{\mu}$ is a stable matching in $\check{\Gamma}$, then $\mu$ is a stable matching in $\Gamma$ where $\mu(f)=\cup _{\check{f}\in Q_f}\check{\mu}(\check{f})$ for each $f\in F$, $\mu(w)=f$ if $\check{\mu}(w)\in Q_f$ for each $w\in W$, and $\mu({\o})=\check{\mu}({\o})$.
\end{lemma}

\begin{proof}
$\mu$ is a matching since $\mu(w)=f$ if and only if $w\in \mu(f)$. The individual rationalities for workers in $\Gamma$ follow that $\check{\mu}(w)\succ_w{\o}$ implies $\mu(w)\succ_w{\o}$ for each $w\in W$. For each $f\in F$ and each $\check{f}\in Q_f$, since $\check{\mu}(\check{f})=Ch_{\check{f}}(\check{\mu}(\check{f}))=Ch_f(\check{\mu}(\check{f}))\cap W^{\check{f}}=Ch_f(\check{\mu}(\check{f}))$ and $f$ has a complementary preference, we know that $Ch_f(\cup _{\check{f}\in Q_f}\check{\mu}(\check{f}))=\cup _{\check{f}\in Q_f}\check{\mu}(\check{f})$. Hence, $\mu$ is individually rational for firms.

Suppose $f$ and $S\subseteq W$ block $\mu$ in $\Gamma$, then $Ch_f(S\cup\mu(f))\neq\mu(f)$. Since $\mu$ is individually rational for firms, there is at least a worker $w^*$ such that $w^*\in Ch_f(S\cup\mu(f))$ and $w^*\notin \mu(f)$. Let $f^*\in \check{F}$ be the firm such that $f^*\in Q_f$ and $w^*\in W^{f^*}$. We then prove that $f^*$ and $S'=Ch_f(S\cup\mu(f))\cap W^{f^*}$ form a blocking coalition in $\check{\mu}$.

Since $S$ and $f$ block $\mu$, $f\succeq_w\mu(w)$ for each $w\in S\cup\mu(f)$. Then, since $S'\subseteq S\cup\mu(f)$, we know that $f^*\check{\succeq}_w\check{\mu}(w)$ for each $w\in S'$. On the other hand, we have

\begin{equation}\label{Set}
\begin{aligned}
S'&=Ch_f((S\cap W^{f^*})\cup(\mu(f)\cap W^{f^*}))\\
&=Ch_f((S\cap W^{f^*})\cup\check{\mu}(f^*))
\end{aligned}
\end{equation}

The first equality is because $f$'s choice over $W^{f^*}$ is independent of other workers outside this set. We then have
\begin{align*}
Ch_{f^*}(S'\cup\check{\mu}(f^*))&=Ch_{f}(S'\cup\check{\mu}(f^*))\cap W^{f^*}\\
&=Ch_{f}(S'\cup\check{\mu}(f^*))\\
&=S'
\end{align*}
The first equality is due to the definition of $Ch_{f^*}$.
The second equality is because $S'\subseteq W^{f^*}$ and $\check{\mu}(f^*)\subseteq W^{f^*}$. The third equality is due to (\ref{Set}) and $S'\subseteq(S'\cup\check{\mu}(f^*))\subseteq((S\cap W^{f^*})\cup\check{\mu}(f^*))$.\footnote{(\ref{Set}) implies $S'\subseteq((S\cap W^{f^*})\cup\check{\mu}(f^*))$, and thus $(S'\cup\check{\mu}(f^*))\subseteq((S\cap W^{f^*})\cup\check{\mu}(f^*))$.}

$w^*\notin \mu(f)$ implies $w^*\notin \check{\mu}(f^*)$, then since $w^*\in S'$, we have $S'\neq \check{\mu}(f^*)$. Thus, $S'\check{\succ}_{f^*}\check{\mu}(f^*)$.
\end{proof}

\subsection{Proof of Theorem \ref{thm_neighbour}}\label{proof_thm4}

Suppose there exists a cycle of length $k\geq3$
\begin{equation*}
(w^1, S^1, w^2, S^2, ..., w^k, S^k, w^1)
\end{equation*}
where $w^i,w^{i+1}\in S^i (i\in\{1,2,\ldots,k-1\})$, $w^k,w^1\in S^k$, and no set of the cycle contains more than two workers of the cycle. We show the nonexistence of such a cycle by reaching a contradiction. We do not require $k$ to be odd, and thus we are proving that elements from $\{W^v|v\in V\}$ form a \emph{totally balanced} matrix. A 0-1 matrix is totally balanced if it does not contain as a submatrix the incidence matrix of any cycle of length at least 3.

Let $\lhd$ be the partial order over the vertices such that $v\lhd v'$ if the unique path from the root to $v'$ passes through $v$. Let $\widetilde{v}^1,\widetilde{v}^2,\ldots,\widetilde{v}^k$ be the vertices of the technology tree such that $w^i$ engages in a neighbour of upgrades from $E_{\widetilde{v}^i\rightarrow}$ for each $i\in\{1,\ldots,k\}$.

$w^1,w^2\in S^1$ implies that one of the following three cases must hold: (i) $\widetilde{v}^1\lhd\widetilde{v}^2$, (ii) $\widetilde{v}^2\lhd\widetilde{v}^1$, or (iii) $\widetilde{v}^1=\widetilde{v}^2$. Then, $w^1\notin S^2$ and $w^2\in S^2$ imply that (i) does not hold; $w^1\in S^k$ and $w^2\notin S^k$ imply that (ii) does not hold. We have $\widetilde{v}^1=\widetilde{v}^2$. Similar arguments hold for any pair $\widetilde{v}^i,\widetilde{v}^{i+1} (i\in\{1,2,\ldots,k-1\})$, and $\widetilde{v}^k,\widetilde{v}^1$. We have $\widetilde{v}^1=\widetilde{v}^2=\ldots=\widetilde{v}^k$. Let $\widetilde{v}=\widetilde{v}^1=\widetilde{v}^2=\ldots=\widetilde{v}^k$.

We then know that there exist $e^1,\ldots,e^k\in E_{\widetilde{v}\rightarrow}$ such that among these upgrades $w^1$ engages in only $e^k$ and $e^1$, and $w^i$ engages in only $e^{i-1}$ and $e^i$ for each $i\in\{2,\ldots,k\}$. Since $w^1$ engages in only a neighbour of upgrades, we know that either (ia) $e^k>_{\widetilde{v}}e^2$, $e^1>_{\widetilde{v}}e^2$ or (ib) $e^2>_{\widetilde{v}}e^k$, $e^2>_{\widetilde{v}}e^1$.

Suppose (ia) holds. Since $w^2$ engages in only a neighbour of upgrades, we know that either (iia) $e^1>_{\widetilde{v}}e^3$, $e^2>_{\widetilde{v}}e^3$ or (iib) $e^3>_{\widetilde{v}}e^1$, $e^3>_{\widetilde{v}}e^2$. If (iib) holds, then $e^1>_{\widetilde{v}}e^2$ and $e^3>_{\widetilde{v}}e^1$ contradict that $w_3$ engages in only a neighbour of upgrades. Thus, (iia) holds. We proceed with the following inductive argument. Suppose $e^s>_{\widetilde{v}}e^{s+2}$ and $e^{s+1}>_{\widetilde{v}}e^{s+2}$ for all $s\in\{1,\ldots,t\} (t\leq k-3)$. Since $w^{t+2}$ engages in only a neighbour of upgrades, we know that either (iiia) $e^{t+1}>_{\widetilde{v}}e^{t+3}$, $e^{t+2}>_{\widetilde{v}}e^{t+3}$ or (iiib) $e^{t+3}>_{\widetilde{v}}e^{t+1}$, $e^{t+3}>_{\widetilde{v}}e^{t+2}$. If (iiib) holds, then $e^{t+1}>_{\widetilde{v}}e^{t+2}$ and $e^{t+3}>_{\widetilde{v}}e^{t+1}$ contradict that $w^{t+3}$ engages in only a neighbour of upgrades. Thus, (iiia) holds. This inductive argument indicates that $e^s>_{\widetilde{v}}e^{s+2}$ and $e^{s+1}>_{\widetilde{v}}e^{s+2}$ for all $s\in\{1,\ldots,k-2\}$, and thus $e^2>_{\widetilde{v}}\cdots>_{\widetilde{v}}e^k$. It contradicts $e^k>_{\widetilde{v}}e^2$ in (ia).

Suppose (ib) holds, we can reach a contradiction with a similar argument as the above paragraph.

\end{document}